\documentclass[12pt,a4paper]{article}

\usepackage[utf8]{inputenc}
\usepackage[T1]{fontenc}
\usepackage{amsmath}
\usepackage{amssymb}
\usepackage{amsthm}
\usepackage{mathtools}
\usepackage{bm}
\usepackage{xcolor}
\usepackage{mathrsfs}
\usepackage{hyperref}

\usepackage[margin=1in]{geometry}
\usepackage{graphicx}
\usepackage{booktabs}
\usepackage{caption}
\usepackage{subcaption}
\usepackage{microtype}

\usepackage{algorithm}
\usepackage{algpseudocode}

\usepackage[backend=biber, style=numeric, sorting=nyt]{biblatex}
\usepackage{hyperref}
\hypersetup{
    colorlinks=true,
    linkcolor=blue,
    citecolor=red,
    urlcolor=teal,
}

\newtheorem{theorem}{Theorem}[section]
\newtheorem{lemma}[theorem]{Lemma}

\newtheorem{example}{Example}

\begin{document}

\title{Two-dimensional constacyclic codes over finite chain rings\footnotetext{2020 Mathematics Subject Classication. Primary: 11T71; Secondary: 94B05.\\
Keywords: Primitive idempotents, finite chain rings, two-dimensional constacyclic codes, MHDR codes.\\
*Corresponding author: Ridhima Thakral.}}
\author{
Vaishali Singh$^{1}$, Sucheta Dutt$^{2}$ and Ridhima Thakral$^{3,*}$\\[6pt]
\href{mailto:vaishali.phd24maths@pec.edu.in}{vaishali.phd24maths@pec.edu.in}\\[6pt]
\href{mailto:sucheta@pec.edu.in}{sucheta@pec.edu.in}\\[6pt]
\href{mailto:ridhimathakral@pec.edu.in}{ridhimathakral@pec.edu.in}\\[6pt]
$^{1,2,3}$Department of Mathematics,\\ 
Punjab Engineering College 
(Deemed to be University),\\
Sector 12, Chandigarh 160012, India\\
}
\date{}
\maketitle

\begin{abstract}
The main focus of this paper is on the algebraic structure of two-dimensional $(\lambda,\mu)$-constacyclic codes of length $\ell\mathrm{m}$ over finite chain rings with residue field $\mathbb{F}_q$, where $q \equiv 1 \pmod{r\mathrm{m}}$ and $r$ denotes the multiplicative order of $\bar{\mu}$.  In this paper, the structure of two-dimensional $(\lambda,\mu)$-constacyclic codes is obtained. Our approach relies on analysing primitive idempotents within the finite chain ring to determine the generators of these codes. We also find the condition under which two-dimensional constacyclic codes are maximum hamming distance with respect to rank (MHDR) over finite chain rings.

\end{abstract}

\section{Introduction}
The practical need for reliable data transmission and storage in digital systems has motivated a deep and extensive study of error-correcting codes. The algebraic framework of cyclic and constacyclic codes has made them one of the most extensively studied classes of linear error-correcting codes. Their structural properties enable efficient encoding and decoding procedures, making them suitable for a wide range of communication and storage applications. Therefore, extensive research has been conducted on the structural properties and generator characterisations of these codes over fields and rings. The exploration of codes over finite chain rings has gained prominence as finite chain rings form a superset of finite fields. Later, two-dimensional cyclic and constacyclic codes developed as a natural generalisation of classical cyclic codes over finite fields as well as finite chain rings.

\medskip
In 1977, Imai \cite{imai1977theory} was the first to propose the idea of two-dimensional binary cyclic codes and has since attracted significant attention due to its theoretical richness and practical relevance in multi-dimensional signal processing and data transmission. In 2012, C.Güneri and F. Özbudak \cite{guneri2012relation} showed that two-dimensional cyclic codes are algebraically equivalent to a special class of quasi-cyclic codes. In 2016, a study focused on two-dimensional cyclic codes of length $s.2^k$ and their dual codes over a finite field was carried out by Z. Sepasdar and K. Khashyarmanesh \cite{sepasdar2016characterizations}. In 2017, Z. Sepasdar \cite{sepasdar2017generator} also derived the generator matrix corresponding to two-dimensional cyclic codes of any given length over finite fields. An algebraic structure of some two-dimensional constacyclic codes of length $2p^s.2^k$ was given by Z. Rajabi and K. Khashyarmanesh \cite{rajabi2018repeated} in 2017. In 2021, S.Bhardwaj and M.Raka \cite{bhardwaj2022} analyzed $(\alpha,\beta)$-constacyclic codes of length $sl$ over $\mathbb{F}_q$. In 2023, D.Garg and S.Dutt \cite{garg2024two} obtained the generators of two-dimensional cyclic codes over finite chain rings using primitive idempotents.

\medskip
Motivated by these developments, the present work focuses on the construction of two-dimensional $(\lambda, \mu)$-constacyclic codes of length $\ell\mathrm{m}$ over finite chain rings with residue field $\mathbb{F}_q$, where $q \equiv 1 \pmod{r\mathrm{m}}$ and $r$ denotes the multiplicative order of $\bar{\mu}$.

The remainder of this paper is structured as follows. Necessary definitions and foundational results related to constacyclic codes over finite chain rings are given in section~2. Generators of two-dimensional $(\lambda,\mu)$-constacyclic codes of length $\ell\mathrm{m}$ over finite chain rings with the help of primitive idempotents are obtained in section~3.  A condition for two-dimensional constacyclic codes to be MHDR over a finite chain ring is obtained in section~4. Finally, concluding remarks of the paper are given in section~5.

\section{Preliminaries}

Let $\mathcal{R}$ be a finite commutative ring. A code of length $\ell$ over $\mathcal{R}$ is a nonempty subset of $\mathcal{R}^{\ell}$. A linear code $\mathcal{C}$ of length $\ell$ over $\mathcal{R}$ is an $\mathcal{R}$-submodule of $\mathcal{R}^{\ell}$. Let $\lambda$ be a unit element in $\mathcal{R}$. A linear code $\mathcal{C}$ of length $\ell$ over $\mathcal{R}$ is said to be a $\lambda$-constacyclic code if
$(\eta_0, \eta_1, \dots, \eta_{\ell-1}) \in \mathcal{C}
\;\Rightarrow\;
(\lambda \eta_{\ell-1}, \eta_0, \eta_1, \dots, \eta_{\ell-2}) \in \mathcal{C}$. A $\lambda$-constacyclic code $\mathcal{C}$ is an ideal of the quotient ring $\mathcal{R}[x]/\langle x^{\ell} - \lambda \rangle$ under the map $\psi : \mathcal{R}^{\ell} \to \mathcal{R}[x]/\langle x^{\ell} - \lambda \rangle $ defined as $\psi(\eta_0,\eta_1, \dots, \eta_{\ell-1})=\eta_0+\eta_1x+ \dots+\eta_{\ell-1}x^{\ell-1} \pmod{x^{\ell}-\lambda}$.
For $\lambda=1$, the code is known as a cyclic code.

A nonzero element $\omega \in \mathcal{R}$ is called a primitive $n^{th}$ root of unity if $n$ is the least positive integer for which $\omega^{n} = 1$. If an element $e \in \mathcal{R}$ satisfies the property $e^2=e$, then $e$ is called an idempotent element of $\mathcal{R}$. In addition, in a commutative ring $\mathcal{R}$ with identity, a set of idempotent elements is said to be primitive idempotents when the elements are mutually orthogonal, and their sum equals the identity of $\mathcal{R}$.

Let $\gamma \in \mathcal{R}$, such that $\langle \gamma \rangle$ is the unique maximal ideal of $\mathcal{R}$ generated by $\gamma$. If all ideals of a finite commutative ring $\mathcal{R}$ can be arranged in a single chain under the inclusion relation, then $\mathcal{R}$ is termed as a finite chain ring. Equivalently, $\langle 0 \rangle = \langle \gamma^{\rho} \rangle \subset \langle \gamma^{\rho-1} \rangle \subset \cdots \subset \langle \gamma \rangle \subset \langle \gamma^{0} \rangle = \mathcal{R}$, where $\rho$ is the nilpotency index of $\gamma$. In this case, $\mathcal{R}$ forms a principal ideal ring with $\langle \gamma \rangle$ as its unique maximal ideal. The factor ring $\mathcal{\overline{R}}=\mathcal{R}/\langle \gamma\rangle$ forms a finite field, say $\mathbb{F}_{q}$.The ring homomorphism $\mathcal{R} \to \mathbb{F}_{q}$, defined by $a \mapsto \overline{a}$, extends naturally coefficientwise from $\mathcal{R}[x]$ onto $\mathbb{F}_{q}[x]$. 

\medskip
Let $\mathcal{C}$ be a $\lambda$-constacyclic code of length $\ell$ over $\mathcal{R}$. For each integer $i$ with $0 \le i \le \rho-1$, the $i$-th torsion code of $\mathcal{C}$, denoted by $\mathrm{Tor}_i(\mathcal{C})$, is defined as 
$\mathrm{Tor}_i(\mathcal{C})=\{\,\overline{c(x)}\in \overline{\mathcal{R}}[x]/\langle x^{\ell}- \overline{\lambda} \rangle \mid \gamma^{\,i}c(x) \in \mathcal{C} \,\}$. It is easy to see that $\mathrm{Tor}_i(\mathcal{C})$ is a $\overline{\lambda}$-constacyclic code of length $\ell$ over $\mathbb{F}_q$.

\medskip
A two-dimensional code of length $\ell\mathrm{m}$ over $\mathcal{R}$ is a nonempty subset of $\mathcal{R}^{\ell\mathrm{m}}$ whose codewords are regarded as an $\ell\times \mathrm{m}$ array defined as $\eta=\big(\eta_{i,j}\big),0\le i\le \ell-1,0\le j\le \mathrm{m}-1,\eta_{i,j}\in \mathcal{R}$. Each codeword $\eta=(\eta_{i,j})$ over $\mathcal{R}$ in a two-dimensional code can be uniquely represented by a bivariate polynomial
\[
\eta(x,y)=\sum_{j=0}^{\mathrm{m}-1}\left(\sum_{i=0}^{\ell-1} \eta_{i,j}x^{i}\right)y^{j}
\;\in\; \mathcal{R}[x,y].
\]

Let $\lambda,\mu\in \mathcal{R}^*$ be units of $\mathcal{R}$. We define two shift operators on such arrays as follows. The \emph{row $\lambda$–constacyclic shift} is given by 
\[
\tau_{\lambda}(\eta) =
\begin{pmatrix}
\lambda\,\eta_{\ell-1,0} & \lambda\,\eta_{\ell-1,1} & \cdots & \lambda\,\eta_{\ell-1,\mathrm{m}-1} \\
\eta_{0,0} & \eta_{0,1} & \cdots & \eta_{0,\mathrm{m}-1} \\
\vdots & \vdots & \ddots & \vdots \\
\eta_{\ell-2,0} & \eta_{\ell-2,1} & \cdots & \eta_{\ell-2,\mathrm{m}-1}
\end{pmatrix}
\]
The \emph{column $\mu$–constacyclic shift} is given by
\[
\sigma_{\mu}(\eta) =
\begin{pmatrix}
\mu\,\eta_{0,\mathrm{m}-1} & \eta_{0,0} & \cdots & \eta_{0,\mathrm{m}-2} \\
\mu\,\eta_{1,\mathrm{m}-1} & \eta_{1,0} & \cdots & \eta_{1,\mathrm{m}-2} \\
\vdots & \vdots & \ddots & \vdots \\
\mu\,\eta_{\ell-1,\mathrm{m}-1} & \eta_{\ell-1,0} & \cdots & \eta_{\ell-1,\mathrm{m}-2}
\end{pmatrix}
\]

An $\mathcal{R}$-submodule  $\mathcal{C}$ of $\mathcal{R}^{\ell\mathrm{m}}$ is termed as a \emph{two-dimensional $(\lambda,\mu)$–constacyclic code}
if $\mathcal{C}$ is invariant with respect to both $\tau_{\lambda}$ and $\sigma_{\mu}$, i.e.,
\[
\eta\in C \;\Longrightarrow\; \tau_{\lambda}(\eta)\in C \text{ and } \sigma_{\mu}(\eta)\in C.
\]

A two-dimensional $(\lambda,\mu)$–constacyclic codes over $\mathcal{R}$ is viewed as an ideal of the quotient ring $\mathcal{R}[x,y]\big/\langle x^{\ell}-\lambda,\; y^{\mathrm{m}}-\mu\rangle$. 

\medskip
Let $\mathcal{C}$ be a two-dimensional constacyclic code over the finite chain ring $\mathcal{R}$ of length $\ell \mathrm{m}$, and let $\eta = (\eta_{ij})$ and $\eta' = (\eta'_{ij})$, where $0 \le i \le \ell - 1$ and $0 \le j \le \mathrm{m} - 1$, be two codewords in $\mathcal{C}$. The hamming distance between $\eta$ and $\eta'$ is defined as $d(\eta,\eta') = \sum_{j=0}^{\mathrm{m} - 1} \sum_{i=0}^{\ell - 1} \delta(\eta_{ij},\eta'_{ij})$, where $\delta(\eta_{ij},\eta'_{ij}) = 0$ if $\eta_{ij} = \eta'_{ij}$ and $\delta(\eta_{ij},\eta'_{ij}) = 1$ otherwise. The minimum hamming distance of $\mathcal{C}$ is given by $d(\mathcal{C}) = \min_{\eta \ne \eta'} d(\eta,\eta')$. Further, for a codeword $\eta = (\eta_{ij}) \in \mathcal{C}$, the hamming weight, denoted by $wt(\eta)$, is defined as the number of nonzero entries among the components $\eta_{ij}$. The minimum hamming weight of $\mathcal{C}$ is $wt(\mathcal{C}) = \min_{\eta \in \mathcal{C},\; \eta \ne 0} wt(\eta)$. It follows that $d(\mathcal{C}) = wt(\mathcal{C})$.The rank of $\mathcal{C}$, denoted by $Rank(\mathcal{C})$, is defined as the cardinality of a minimal spanning set of $\mathcal{C}$. The code $\mathcal{C}$ is said to be a maximum hamming distance with respect to rank (MHDR) code if its minimum hamming distance satisfies $d(\mathcal{C}) = \ell \mathrm{m} - Rank(\mathcal{C}) + 1$.

Below are some results that are required later in the paper.

\begin{theorem}\label{thm:first}
\cite{norton2000structure}   Consider a monic polynomial $c(x) \in \mathcal{R}[x]$. If its reduction $\bar{c}(x) \in \mathbb{F}_q[x]$ factors as a product of polynomials $c_1(x),\ldots,c_m(x)$ which are monic and pairwise coprime, then $c(x)$ admits a corresponding factorization into monic polynomials $d_1(x),\ldots,d_m(x)$ which are also coprime in pairs over $\mathcal{R}$.
\end{theorem}

\begin{theorem}\label{thm:second}
\cite{norton2000structure}    If $g(x) \in \mathcal{R}[X]$ is a monic polynomial such that its reduction is square-free, then $g(x)$ has a unique decomposition into monic basic irreducible factors which are coprime.
\end{theorem}

Consider the polynomials $f_0(x), f_1(x), \dots, f_s(x)$ arise through an inductive selection process: at each stage, $f_i(x)$ is chosen to be a polynomial of least degree in $\mathcal{C} \setminus \langle f_0(x), \dots, f_{i-1}(x) \rangle$, and among all such least degree polynomials, one whose leading coefficient contains the smallest possible power of $\gamma$, denoted $r_i$. This selection yields a strictly increasing sequence of degrees $n_0 < n_1 < \cdots < n_s$ and a strictly decreasing sequence of powers $r_0 > r_1 > \cdots > r_s$.
The following theorem gives the structure of $\lambda$-constacyclic codes of arbitrary length over finite chain rings.
\begin{theorem}\cite{singh2026structure}\label{thm:cyclicspanning}
    Let $\mathcal{C}$ be a $\lambda$-constacyclic code of length $\ell$ over the finite chain ring $\mathcal{R}$ and $f_i(x); 0 \le i \le s$ be the polynomials as described above, then
    \begin{enumerate}
        \item $\mathcal{C}= \langle f_0(x),f_1(x), \dots , f_s(x) \rangle$. Also $f_i(x)=\gamma^{r_i}h_i(x)$, where $h_i(x)$ is a monic polynomial over the finite commutative chain ring having maximal ideal $\langle \gamma \rangle$ with nilpotency index $\rho -r_i$ and $deg(h_i(x))=n_i$. 
        \item  The minimal spanning set of $\mathcal{C}$ is $\bigcup_{i=0}^{s}\{ f_i(x), xf_i(x), \ldots, x^{n_{i+1}-n_{i}-1}f_i(x)  \}, n_{s+1}=\ell$.
        \item $Rank(\mathcal{C})=\ell-n_0$, where $n_0$ is the degree of polynomial $f_0(x)$.
    \end{enumerate}
\end{theorem}
In the following theorem, generators of two-dimensional constacyclic codes over finite fields are obtained.

\begin{theorem}\label{thm:ccfield}
\cite{bhardwaj2022}    Let $\mathcal{R}=\mathbb{F}_q[x,y]/\langle x^{\ell}-\lambda,\; y^{\mathrm{m}}-\mu\rangle$ and let $\mathcal{C}$ be an ideal of $\mathcal{R}$. For each ${k}=0,1,\ldots,\mathrm{m}-1$, define $\mathcal{I}_{k}=\{\,f(x)\in \mathbb{F}_q[x]/\langle x^{\ell}-\lambda\rangle \mid \mathsf{\zeta}_{k}(y)f(x)\in \mathcal{C}\,\}$. Thus, for each ${k}$, the ideal $\mathcal{I}_{k}$ is a principal and admits a unique monic generator $f_{k}(x)$ dividing $x^{\ell}-\lambda$, and the ideal $\mathcal{C}$ can be written as $\mathcal{C}=\langle \zeta_0(y)f_0(x),\; \zeta_1(y)f_1(x),\; \ldots,\; \zeta_{\mathrm{m}-1}(y)f_{\mathrm{m}-1}(x)\rangle$, such that $\zeta_{k}(y) ; 0 \le {k} \le \mathrm{m}-1$ form the set of primitive idempotents associated with the ring $\mathbb{F}_q[x]/\langle y^{\mathrm{m}} - \mu \rangle$. $Dim(\mathcal{C})=\ell\mathrm{m}- \sum_{k=0}^{\mathrm{m}-1}n_k$, where $n_k=deg(f_k(x))$.
\end{theorem}

\section{Two-dimensional constacyclic codes over finite chain rings}
Let $\mathcal{R}$ be a finite chain ring with residue field $\mathbb{F}_q$. In this section, we derive the generators of two-dimensional $(\lambda,\mu)$-constacyclic codes over $\mathcal{R}$ of length $\ell\mathrm{m}$ with $q\equiv 1 \pmod{r\mathrm{m}}$, where $r$ is multiplicative order of $\bar{\mu}$.
\subsection{Primitive idempotents in a finite chain ring} 

\medskip
Let us begin with the following lemma, in which a set of primitive idempotents of the ring $\mathcal{R}[y]/ \langle y^{\mathrm{m}}- \mu \rangle$ is determined.
\begin{lemma}
Assume that $\mathcal{R}$ is a finite chain ring and $\mu$ is a unit element of $\mathcal{R}.$ Then, the polynomials
\begin{equation*}
    \begin{aligned}
        \psi_k(y)= \prod_{\substack{j=0\\ j\ne k}}^{\mathrm{m}-1} \frac{y - \theta^{1+jr}}{\theta^{1+kr} - \theta^{1+jr}}
\end{aligned}
\end{equation*}
form a set of primitive idempotents associated with the ring $\mathcal{R}[y]/\langle y^{\mathrm{m}} - \mu\rangle$ for $ 0 \le k \le \mathrm{m} - 1$ i.e. $\psi_j(y)\psi_k(y)=0$ whenever $j \ne k$ and $\psi_k^2(y)=\psi_k(y)$ in the ring $\mathcal{R}[y]/\langle y^{\mathrm{m}} - \mu \rangle$. Also $\sum_{k=0}^{\mathrm{m} - 1} \psi_k(y) = 1$.
\end{lemma}

\begin{proof}
Let $\omega \in \mathbb{F}_q$ be a primitive $r\mathrm{m}^{th}$ root of unity such that $\omega^{\mathrm{m}}=\bar{\mu}$. It is easy to see that such an $\omega$ exists. Further, pick an element $\theta \in \mathcal{R}$ such that $\bar{\theta}=\omega$. It can be easily seen that $\theta^{\mathrm{m}}=\mu$. As $\theta$ is primitive $r\mathrm{m}^{th}$ root of unity $\theta,\theta^{1+r},\theta^{1+2r}, \dots, \theta^{1+(\mathrm{m}-1)r}$ are all distinct. Also, all of these satisfy the equation $y^{\mathrm{m}}=\mu$. Therefore, $\theta^{1+kr},~~0\le k \le \mathrm{m}-1$ are all $\mathrm{m}$ roots of the polynomial $y^{\mathrm{m}}-\mu$ of degree $\mathrm{m}$ and therefore, we have the factorization
    \begin{align*}
        y^{\mathrm{m}}-\mu&= (y-\theta)(y-\theta^{1+r})(y-\theta^{1+2r})...(y-\theta^{1+(\mathrm{m}-1)r})\\
                 &=\prod_{k=0}^{\mathrm{m}-1} (y-\theta^{1+kr})
    \end{align*}
Now consider the polynomials
\begin{equation*}
    \begin{aligned}
        \psi_k(y)= \prod_{\substack{j=0\\ j\ne k}}^{\mathrm{m}-1} \frac{y - \theta^{1+jr}}{\theta^{1+kr} - \theta^{1+jr}}~~~~ \text{ for } ~~ 0 \le k \le \mathrm{m} - 1.
\end{aligned}
\end{equation*}
Clearly $\psi_k(\theta^{1+jr})=0 ~~\text{for}~~ 0 \le j \le \mathrm{m}-1, k\ne j$ and $\psi_k(\theta^{1+kr})=1$. Therefore, $\sum_{k=0}^{\mathrm{m} - 1} \psi_k(\theta^{1+kr}) = 1$. Also $y^{\mathrm{m}}-\mu$ is a factor of $\psi_k(y)(\psi_k(y) - 1)$ implies that $\psi_k(y)(\psi_k(y) - 1)\equiv0 \pmod{y^{\mathrm{m}}-\mu}$ and therefore $\psi_k^2(y) = \psi_k(y)$ in $\mathcal{R}[y]/ \langle y^{\mathrm{m}}-\mu \rangle$. Therefore, we can conclude that the polynomials $\psi_0(y),\psi_1(y),...\psi_{\mathrm{m}-1}(y)$ form a set of primitive idempotents in the ring $\mathcal{R}[y]/(y^{\mathrm{m}}-\mu)$.
\end{proof}

\begin{lemma}\label{lem:units}
    $y^j\psi_k(y)=(\theta^{1+kr})^j\psi_k(y)$ ~~ \text{for} ~~ $0\leq k,j\leq \mathrm{m}-1$.
\end{lemma}
\begin{proof}
We establish the result via induction on $j$. 

Consider the polynomial $(y - \theta^{1+kr})\psi_k(y)$. Clearly
$(y - \theta^{1+kr})\psi_k(y) = 0$ in $\mathcal{R}[y]/ \langle y^{\mathrm{m}}-\mu \rangle$.
which implies that $y\psi_k(y) = \theta^{1+kr} \psi_k(y)$, i.e., the claim is valid for $j=1$.

Assume that the result is valid for $j=n$ where $n$ is some positive integer, i.e., $y^n\psi_k(y) = (\theta^{1+kr})^n\psi_k(y)$. Now $y^{n+1}\psi_k(y)= y(y^n\psi_k(y))=y((\theta^{1+kr})^n\psi_k(y))=(\theta^{1+kr})^n(y\psi_k(y))=(\theta^{1+kr})^n(\theta^{1+kr}\psi_k(y))=(\theta^{1+kr})^{n+1}\psi_k(y)$.
Thus, the result holds for $j=n+1$ whenever it holds for $j=n$. Hence, by the principle of mathematical induction, the result is valid for all $j \in \{0, 1, 2, \dots, \mathrm{m}-1\}$.
\end{proof}
Examples illustrating the above results are given below.
\begin{example}\label{ex:third}
Let $\mathcal{R}=\mathbb{Z}_{125}$ with the residue field $\mathbb{F}_{5}$ with nilpotency index $3.$ Consider the ring $\mathcal{R}[y]/ \langle y^4-1 \rangle$, it is clear that $\mathrm{m}=4, \mu=1$ and order of $\mu$ is $1.$ Then $\theta= 57$ is the $4^{th}$ root of unity such that $57^4=1.$ Then the polynomial $y^{4}-1$ factors over $\mathbb{Z}_{125}$ as
\[
\begin{aligned}
y^{4}-1
&=\prod_{k=0}^{3}(y-57^{1+i})\\
&=(y-57)(y-124)(y-68)(y-1).
\end{aligned}
\]
Hence, the primitive idempotents are:
\[
\begin{aligned}
\psi_0(y)&= 94+17y+31y^2+108y^3,\\
\psi_1(y)&= 94+31y+94y^2+31y^3,\\
\psi_2(y)&= 94+108y+31y^2+17y^3,\\
\psi_3(y)&= 94+94y+94y^2+94y^3.
\end{aligned}
\]

\end{example}

\begin{example}\label{ex:fourth}
Let $\mathcal{R}=\mathbb{Z}_{169}$ with the residue field $\mathbb{F}_{13}$ with nilpotency index $2.$ Consider the ring $\mathcal{R}[y]/ \langle y^6-168 \rangle$, it is clear that $\mathrm{m}=6, \mu=168$ and order of $\mu$ is $2.$ Then $\theta= 89$ is the $12^{th}$ root of unity such that $89^6=168.$ Then the polynomial $y^{6}-168$ factors over $\mathbb{Z}_{169}$ as
\[
\begin{aligned}
y^{6}-168
&=(y-89)(y-89^3)(y-89^5)(y-89^7)(y-89^9)(y-89^{11})\\
&=(y-89)(y-70)(y-150)(y-80)(y-99)(y-19).
\end{aligned}
\]
Hence, the primitive idempotents are:
\[
\begin{aligned}
\psi_0(y)&= 141+144y+32y^2+101y^3+60y^4+126y^5,\\
\psi_1(y)&= 141+101y+28y^2+68y^3+141y^4+101y^5,\\
\psi_2(y)&= 141+126y+109y^2+101y^3+137y^4+144y^5,\\
\psi_3(y)&= 141+25y+32y^2+68y^3+60y^4+43y^5,\\
\psi_4(y)&= 141+68y+28y^2+101y^3+141y^4+68y^5,\\
\psi_5(y)&= 141+43y+109y^2+68y^3+137y^4+25y^5.
\end{aligned}
\]
\end{example}

\subsection{Generators of two-dimensional constacyclic codes over a finite chain ring}
Consider a two-dimensional $(\lambda,\mu)$-constacyclic code $\mathscr{C}$ having length $\ell\mathrm{m}$ over a finite chain ring $\mathcal{R}$. The code $\mathscr{C}$ can be expressed as an ideal of the ring $\mathcal{R}[x,y]/ \left\langle x^{\ell} - \lambda,y^{\mathrm{m}} - \mu \right\rangle$ and vice versa. Define the sets
\[
\begin{aligned}
J_k &= \{\, f(x) \in \mathcal{R}[x]/\langle x^{\ell} - \lambda\rangle \; | \; \psi_k(y)f(x) \in \mathscr{C} \,\},~~~~~ \text{for} ~~ 0 \le k \le \mathrm{m}-1
\end{aligned}
\]
Clearly $J_k,0 \le k \le \mathrm{m}-1$ are ideals of the ring $\mathcal{R}[x]/\left\langle x^{\ell}-\lambda\right\rangle$. Therefore, $J_k,0 \le k \le \mathrm{m}-1$ are $\lambda$-constacyclic codes of length $\ell$ over the ring $\mathcal{R}$.  By theorem~\ref{thm:cyclicspanning}, each $J_k$ is generated by the polynomials $f_0^{(k)}(x),f_1^{(k)}(x),...,f_{s_{k}}^{(k)}(x)$ such that $deg(f_i^{(k)}(x))=n_i^{(k)}$ for $0 \le k \le \mathrm{m}-1$. 

In the following theorem, we determine the generators of two-dimensional $(\lambda,\mu)$-constacyclic codes of length $\ell\mathrm{m}$ over $\mathcal{R}$. 
\begin{theorem}\label{thm:structure}
Let $\mathscr{C}$ be a two-dimensional $(\lambda,\mu)$-constacyclic code of length $\ell\mathrm{m}$ over a finite chain ring $\mathcal{R}$ that has a residue field $F_q$, satisfying $q \equiv 1 \pmod{r\mathrm{m}}$, where $r$ is multiplicative order of $\bar{\mu}$. Then, the set of generators for the code $\mathscr{C}$ is given by 
\[
\left\{\, \psi_k(y)\, f_i^{(k)}(x) \;\middle|\; 0 \le i \le s_{k},\; 0 \le k \le \mathrm{m}-1 \,\right\},
\]
where $\psi_k(y)$ are the primitive idempotents of the ring $\mathcal{R}[y]/\langle y^{\mathrm{m}}-\mu \rangle$ for $0 \le k \le \mathrm{m}-1$ and the polynomials $f_i^{(k)}(x)$, $0 \le i \le s_{k}$, generate the $\lambda$-constacyclic code $J_k = \left\{\, f_i^{(k)}(x) \in \mathcal{R}[x]/\langle x^{\ell}-\lambda \rangle \;\middle|\; \psi_k(y) f_i^{(k)}(x) \in \mathscr{C} \,\right\}.$
\end{theorem}
\begin{proof}
Consider an arbitrary element $h(x,y)\in \mathscr{C}$. It can be expressed as follows:
\begin{align*}
    h(x,y)=\sum_{j=0}^{\mathrm{m}-1}h_j(x)y^j,
\end{align*}
where each $h_j(x)\in \mathcal{R}[x]/\langle x^{\ell}-\lambda\rangle$ for $0 \le j \le \mathrm{m}-1$. Then
\begin{align*}
h(x,y)\psi_k(y)&=\Big(\sum_{j=0}^{\mathrm{m}-1}h_j(x)y^j\Big)\psi_k(y)\\
&=h_0(x)\psi_k(y)+h_1(x)y\psi_k(y)+\cdots+h_{\mathrm{m}-1}(x)y^{\mathrm{m}-1}\psi_k(y)
\end{align*}
Using lemma~\ref{lem:units},
\begin{align*}
h(x,y)\psi_k(y)=h_0(x)\psi_k(y)+h_1(x)(\theta^{1+kr})\psi_k(y)+\cdots+h_{\mathrm{m}-1}(x)(\theta^{1+kr})^{\mathrm{m}-1}\psi_k(y).
\end{align*}
\begin{align*}
h(x,y)\psi_k(y)= h(x,\theta^{1+kr})\psi_k(y),
\end{align*}
Since $\mathscr{C}$ is an ideal and $h(x,y)\in\mathscr{C}$, we have
$h(x,y)\psi_k(y)\in\mathscr{C}$ and therefore, by the definition of $J_k$, $h(x,\theta^{1+kr}) \in J_k = \langle f_0^{(k)}(x),f_1^{(k)}(x), \dots , f_{s_k}^{(k)}(x) \rangle$. Thus, we can find polynomials $a_i^{(k)}(x)$ in
$\mathcal{R}[x]/\langle x^{\ell}-\lambda\rangle$ satisfying
\begin{align*}
h(x,\theta^{1+kr})=\sum_{i=0}^{\mathrm{m}-1} a_i^{(k)}(x)\,f_i^{(k)}(x).
\end{align*}
Next, using $\sum_{k=0}^{\mathrm{m}-1}\psi_k(y)=1$,
\begin{align*}
h(x,y) &=h(x,y)\cdot 1 \\
&=h(x,y)\sum_{k=0}^{\mathrm{m}-1}\psi_k(y)\\
&=\sum_{k=0}^{\mathrm{m}-1} h(x,y)\psi_k(y)\\
&=\sum_{k=0}^{\mathrm{m}-1} h(x,\theta^{1+kr})\psi_k(y)\\
&= \sum_{k=0}^{\mathrm{m}-1} \sum_{i=0}^{s_k} a_i^{(k)}(x)f_i^{(k)}(x)\psi_k(y).
\end{align*}
Thus, every $h(x,y)\in\mathscr{C}$ is generated by $\Big\{\psi_k(y)f_i^{(k)}(x) \mid 0 \le i \le s_k, \ 0\le k\le \mathrm{m}-1\Big\}$.
Hence,
\[
\mathscr{C}\subseteq
\left\langle\, \psi_k(y)f_i^{(k)}(x) \mid 0 \le i \le s_k, \ 0\le k\le \mathrm{m}-1\right\rangle.
\]
Also note that for each $k$ and each generator $f_i^{(k)}(x)\in J_k$,
we have $\psi_k(y)f_i^{(k)}(x)\in \mathscr{C}$ by the definition of $J_k$.
Therefore, the ideal generated by these elements is contained in $\mathscr{C}$. Hence,
\[
\mathscr{C} =
\left\langle\, \psi_k(y)f_i^{(k)}(x) \mid 0 \le i \le s_k, \ 0\le k\le \mathrm{m}-1\right\rangle.
\]

\end{proof}

\begin{theorem}\label{thm:rank}
        Let $\mathscr{C}$ be a two-dimensional constacyclic code of length $\ell\mathrm{m}$ over $\mathcal{R}$ that has a residue field $F_q$, satisfying $q \equiv 1 \pmod{r\mathrm{m}}$, where $r$ is multiplicative order of $\bar{\mu}$. Then the minimal spanning set of $\mathscr{C}$ is given by $B=\bigcup_{k=0}^{\mathrm{m}-1}\bigcup_{i=0}^{s_k}\{\psi_k(y)f_i^{(k)}(x), x\psi_k(y)f_i^{(k)}(x),  \dots, \\ x^{n_{i+1}^{(k)}-n_i^{(k)}-1}\psi_k(y)f_i^{(k)}(x) \}$ where $n_i^{(k)}=deg(f_i^{(k)}(x))$ and ${n_{s_k+1}^{(k)}}={\ell}$. $Rank(\mathscr{C})= \ell\mathrm{m}-\sum_{k=0}^{\mathrm{m}-1}n_0^{(k)}$.
\end{theorem}
\begin{proof}
    Choose an arbitrary element $h(x,y)$ of $\mathscr{C}$. From theorem~\ref{thm:structure}, we have\\ $h(x,y)= \sum_{k=0}^{\mathrm{m}-1}h(x,\theta^{1+kr})\psi_k(y)$ where $h(x,\theta^{1+kr}) \in J_k$. By theorem~\ref{thm:cyclicspanning}, $J_k$ has a minimal spanning set $B_k= \bigcup_{i=0}^{s_k}\{ f_i^{(k)}(x), xf_i^{(k)}(x), \dots , x^{n_{i+1}-n_i-1}f_i^{(k)}(x)\}, n_{s_k+1}= \ell$. Therefore, $h(x,\theta^{1+kr}) \in Span(B_k)$, it shows that $h(x,y) \in Span(B)$. It remains only to establish that none of the elements of $B$ arises as a linear combination of the remaining ones. If possible, let there exist such an element, i.e.
\[
\begin{aligned}
x^{\ell-n^{(k)}_{s_k}-1} f^{(k)}_{s_k}(x)\psi_k(y)
&= \sum_{k=0}^{\mathrm{m}-1}\Bigg[
   \sum_{i=0}^{s_k-1}\Big(
   a^{(k)}_{i,0} f^{(k)}_i(x)\psi_k(y)
   + a^{(k)}_{i,1} x f^{(k)}_i(x)\psi_k(y)
   + \cdots \\
&\qquad
   + a^{(k)}_{i,\,n^{(k)}_{i+1}-n^{(k)}_i-1}
     x^{\,n^{(k)}_{i+1}-n^{(k)}_i-1}
     f^{(k)}_i(x)\psi_k(y)
   \Big) \\
&\qquad
   + a^{(k)}_{s_k,0} f^{(k)}_{s_k}(x)\psi_k(y)
   + \cdots +a^{(k)}_{s_k,\,\ell-n^{(k)}_{s_k}-2}\,
x^{\,\ell-n^{(k)}_{s_k}-2}\,f^{(k)}_{s_k}(x)\,\psi_k(y) \Bigg] \\
&= \sum_{k=0}^{\mathrm{m}-1}\sum_{i=0}^{s_k}
   a^{(k)}_i(x) f^{(k)}_i(x)\psi_k(y).
\end{aligned}
\]
where $
a^{(k)}_{i}(x)
=
a^{(k)}_{i,0}+a^{(k)}_{i,1}x+\cdots
+a^{(k)}_{i,\,n^{(k)}_{i+1}-n^{(k)}_{i}-1}\,x^{\,n^{(k)}_{i+1}-n^{(k)}_{i}-1},
\quad 0\le i\le s_k-1$ and $
a^{(k)}_{s_k}(x)
=
a^{(k)}_{s_k,0}+a^{(k)}_{s_k,1}x+\cdots
+a^{(k)}_{s_k,\,\ell-n^{(k)}_{s_k}-2}\,x^{\,\ell-n^{(k)}_{s_k}-2}.$ Substituting $y=\theta^{1+kr}$, using $\psi_k(\theta^{1+kr})=1 ~~\text{and  } \psi_j(\theta^{1+kr})=0 ~~\text{where  } j \ne k$. It can be seen that the degree of the L.H.S becomes $\ell-1$ and that of the R.H.S. is at most $\ell-2$, which is a contradiction. The same argument can be used to prove that any other element $x^{\,n^{(k)}_{i+1}-n^{(k)}_i-1}f^{(k)}_i(x), 0 \le i \le s_k-1$ cannot be represented as a linear combination of other elements in the set $B$. Therefore, $B$ is a minimal spanning set. Hence $Rank(\mathscr{C})=\ell\mathrm{m}-\sum_{k=0}^{\mathrm{m}-1}n_0^{(k)}$.
\end{proof}
Examples illustrating the above results are given below.
\begin{example}\label{ex:fifth}
    Consider $\mathcal{R}=Z_{125}$ with the maximal ideal $\langle\gamma\rangle=\langle5\rangle$ and nilpotency index $\rho=3$. The corresponding residue field is $\mathbb{F}_5$. Let $\mathscr{C}$ be a two-dimensional $(2,1)$-constacyclic code of length $5 \times 4$ over $Z_{125}$ i.e $\mathscr{C}$ is an ideal of $\mathcal{R}[x,y]/\langle x^5-2, y^4-1 \rangle$. By example~\ref{ex:third}, the set of primitive idempotents of $\mathcal{R}[y]/\langle y^4-1 \rangle$ are
    \[
    \begin{aligned}
    \psi_0(y)&= 94+17y+31y^2+108y^3,\\
    \psi_1(y)&= 94+31y+94y^2+31y^3,\\
    \psi_2(y)&= 94+108y+31y^2+17y^3,\\
    \psi_3(y)&= 94+94y+94y^2+94y^3.
    \end{aligned}
    \]
    By theorem~\ref{thm:cyclicspanning}, following are $2$-constacyclic codes of length $5$ over $\mathcal{R}$,
    \[
    \begin{aligned}
    J_0 &= \langle f^{(0)}_0(x) \rangle = \langle 5(x-2) \rangle, \\
    J_1 &= \langle f^{(1)}_0(x), f^{(1)}_1(x) \rangle = \langle 25(x-2), 5(x-2)^3 \rangle, \\
    J_2 &= \langle f^{(2)}_0(x), f^{(2)}_1(x) \rangle = \langle 5(x-2)^2, (x-2)^4 \rangle, \\
    J_3 &= \langle f^{(3)}_0(x), f^{(3)}_1(x), f^{(3)}_2(x) \rangle 
    = \langle 25(x-2), 5(x-2)^3, (x-2)^4 \rangle.
    \end{aligned}
    \]
    By theorem~\ref{thm:structure}, the set $\{ \psi_0(y)f^{(0)}_0(x),\psi_1(y)f^{(1)}_0(x),\psi_1(y)f^{(1)}_1(x),\psi_2(y)f^{(2)}_0(x),\psi_2(y)f^{(2)}_1(x),\\ \psi_3(y)f^{(3)}_0(x),\psi_3(y)f^{(3)}_1(x),\psi_3(y)f^{(3)}_2(x)\}$ generates a two-dimensional $(2,1)$-constacyclic code of length $20$ over $\mathcal{R}=Z_{125}$ with $Rank(\mathcal{C})=15$.
\end{example}

\begin{example}\label{ex:sixth}
    Consider $\mathcal{R}=Z_{169}$ with the maximal ideal $\langle\gamma\rangle=\langle13\rangle$ and nilpotency index $\rho=2$. The corresponding residue field is $\mathbb{F}_{13}$. Let $\mathscr{C}$ be a two-dimensional $(7,168)$-constacyclic code of length $15 \times 6$ over $Z_{169}$ i.e $\mathscr{C}$ is an ideal of $\mathcal{R}[x,y]/\langle x^{15}-7, y^6-168 \rangle$. By example~\ref{ex:fourth}, the set of primitive idempotents of $\mathcal{R}[y]/\langle y^6-168 \rangle$ are
    \[
    \begin{aligned}
    \psi_0(y)&= 141+144y+32y^2+101y^3+60y^4+126y^5,\\
    \psi_1(y)&= 141+101y+28y^2+68y^3+141y^4+101y^5,\\
    \psi_2(y)&= 141+126y+109y^2+101y^3+137y^4+144y^5,\\
    \psi_3(y)&= 141+25y+32y^2+68y^3+60y^4+43y^5,\\
    \psi_4(y)&= 141+68y+28y^2+101y^3+141y^4+68y^5,\\
    \psi_5(y)&= 141+43y+109y^2+68y^3+137y^4+25y^5.
    \end{aligned}
    \]
    By theorem~\ref{thm:cyclicspanning}, following are $7$-constacyclic codes of length $15$ over $\mathcal{R}$,
    \[
    \begin{aligned}
    J_0 &= \langle f^{(0)}_0(x) \rangle = \langle (x+1) \rangle, \\
    J_1 &= \langle f^{(1)}_0(x) \rangle = \langle 13(x+1) \rangle, \\
    J_2 &= \langle f^{(2)}_0(x) \rangle = \langle x^3+2\rangle, \\
    J_3 &= \langle f^{(3)}_0(x), f^{(3)}_1(x) \rangle 
    = \langle 13,x^3+2 \rangle, \\
    J_4 &= \langle f^{(4)}_0(x), f^{(4)}_1(x) \rangle 
    = \langle 13,x^3+2+13(x+1) \rangle, \\
    J_5 &= \langle f^{(5)}_0(x), f^{(5)}_1(x) \rangle 
    = \langle 13,x^{12}-2x^9+4x^6+5x^3+3 \rangle.
    \end{aligned}
    \]
    By theorem~\ref{thm:structure}, the set $\{ \psi_0(y)f^{(0)}_0(x),\psi_1(y)f^{(1)}_0(x),\psi_2(y)f^{(2)}_0(x), \psi_3(y)f^{(3)}_0(x),\psi_3(y)f^{(3)}_1(x),\\ \psi_4(y)f^{(4)}_0(x),\psi_4(y)f^{(4)}_1(x),\psi_4(y)f^{(5)}_0(x),\psi_4(y)f^{(5)}_1(x)\}$ generates a two-dimensional $(7,168)$-constacyclic code of length $90$ over $\mathcal{R}=Z_{169}$ with $Rank(\mathcal{C})=85$.
\end{example}

\begin{example}\label{ex:seventh}
    Consider $\mathcal{R}=Z_{343}$ with the maximal ideal $\langle\gamma\rangle=\langle7\rangle$ and nilpotency index $\rho=3$. The corresponding residue field is $\mathbb{F}_7$. Let $\mathscr{C}$ be a two-dimensional $(5,18)$-constacyclic code of length $12 \times 2$ over $Z_{343}$ i.e $\mathscr{C}$ is an ideal of $\mathcal{R}[x,y]/\langle x^{12}-5, y^2-18 \rangle$. It is easy to see that $\theta=19$. Therefore, the set of primitive idempotents of $\mathcal{R}[y]/\langle y^2-18 \rangle$ are
    \[
    \begin{aligned}
    \psi_0(y) &= 172 + 334y \\
    \psi_1(y) &= 172 + 9y
    \end{aligned}
    \]
    By theorem~\ref{thm:cyclicspanning}, following are $5$-constacyclic codes of length $12$ over $\mathcal{R}$,
    \[
    \begin{aligned}
    J_0 &= \langle f^{(0)}_0(x),f^{(0)}_1(x) \rangle = \langle 49,(x^6+x^3+4) \rangle, \\
    J_1 &= \langle f^{(1)}_0(x), f^{(1)}_1(x) \rangle = \langle 49,(x^6-x^3+4) \rangle.
    \end{aligned}
    \]
    By theorem~\ref{thm:structure}, the set $\{ \psi_0(y)f^{(0)}_0(x),\psi_0(y)f^{(0)}_1(x),\psi_1(y)f^{(1)}_0(x),\psi_1(y)f^{(1)}_1(x)\}$ generates a two-dimensional $(2,1)$-constacyclic code of length $24$ over $\mathcal{R}=Z_{343}$ with $Rank(\mathcal{C})=24$.
\end{example}

\section{MHDR two-dimensional constacyclic codes}
In this section, we find a condition for two-dimensional constacyclic codes to be MHDR over a finite chain ring.
Let $\mathscr{C}$ be a two-dimensional constacyclic code as defined in theorem~\ref{thm:structure}. Define the set $\mathscr{C}_{\rho-1}= \{h(x,y) \in \mathbb{F}_q[x,y] / \langle x^{\ell}-\overline{\lambda}, y^{\mathrm{m}}-\overline{\mu} \rangle \mid \gamma^{\rho-1}h(x,y) \in \mathscr{C} \}$. It is easy to see that $\mathscr{C}_{\rho-1}$ is an ideal in the ring $\mathbb{F}_q[x,y] / \langle x^{\ell}-\overline{\lambda}, y^{\mathrm{m}}-\overline{\mu} \rangle$, therefore a two-dimensional constacyclic code of length $\ell\mathrm{m}$ over the finite field $\mathbb{F}_q$. By theorem~\ref{thm:ccfield}, $\mathscr{C}_{\rho-1}$ is generated by the set $\{ \zeta_k(y)g_k(x) \mid 0 \le k \le \mathrm{m}-1 \}$ where $\zeta_k(y), 0 \le k \le \mathrm{m}-1$ are primitive idempotents of the ring $\mathbb{F}_q[y] / \langle y^{\mathrm{m}}-\overline{\mu} \rangle$ and $g_k(x), 0 \le k \le \mathrm{m}-1$ are generators of constacyclic code $\{ f(x) \in \mathbb{F}_q[x] / \langle x^{\ell}-\overline{\lambda} \rangle \mid \zeta_k(y)f(x) \in \mathscr{C}_{\rho-1} \}$. Now it is easy to see that $\zeta_k(y)= \overline{\psi_k(y)}$ where $\psi_k(y)$ are primitive idempotents of the ring $\mathcal{R}[x]/ \langle x^{\ell}- \lambda \rangle$ and polynomial $g_k(x)=\overline{f_0^{(k)}(x)}$ since $\gamma^{\delta}f_0^{(k)}(x)$ is the minimum degree polynomial in $J_k$. Therefore, $\mathscr{C}_{\rho-1}= \langle \overline{\psi_k(y)f_0^{(k)}(x)} \mid 0 \le k \le \mathrm{m}-1 \rangle$. Also from theorem~\ref{thm:ccfield} $Dim(\mathscr{C}_{\rho-1})=\ell\mathrm{m}- \sum_{k=0}^{\mathrm{m}-1}n_0^{(k)}$.

\begin{theorem}\label{thm:mhdr}
    Let $\mathscr{C}$ be a two-dimensional $(\lambda,\mu)$-constacyclic code over the finite chain $\mathcal{R}$. Then $\mathscr{C}$ is MHDR over $\mathcal{R}$ if and only if $\mathscr{C}_{\rho-1}$ is MHDR over $\mathbb{F}_q$.
\end{theorem}
\begin{proof}
    From the above and theorem~\ref{thm:rank}, it is clear that $Dim(\mathscr{C}_{\rho-1})= Rank(\mathscr{C})$. It is sufficient to prove $d(\mathscr{C}_{\rho-1})=d(\mathscr{C})$. Let $h(x,y)$ be an element of $\mathscr{C}_{\rho-1}$ such that $wt(\mathscr{C}_{\rho-1})=wt(h(x,y)).$ It is clear that $wt(h(x,y))=wt(\gamma^{\rho-1}h(x,y))$ since $h(x,y) \in \mathbb{F}_q[x,y]$. Also $\gamma^{\rho-1}h(x,y) \in \mathscr{C}$ which implies that $wt(\mathscr{C}_{\rho-1})=wt(h(x,y))=wt(\gamma^{\rho-1}h(x,y)) \geq wt(\mathscr{C}).$ Conversely, let $g(x,y)= \sum_{i=0}^{\rho-1}\gamma^i g_i(x,y) \in \mathscr{C}$ be such that $wt(g(x,y))=wt(\mathscr{C})$. Since $\gamma^{\rho-1}g(x,y) \in \mathscr{C} \implies g_0(x,y) \in \mathscr{C}_{\rho-1}$. Therefore, $wt(\mathscr{C})=wt(g(x,y)) \geq wt(g_0(x,y)) \geq wt(\mathscr{C}_{\rho-1})$. Hence $wt(\mathscr{C})=wt(\mathscr{C}_{\rho-1}) \implies d(\mathscr{C}_{\rho-1})=d(\mathscr{C}).$ Therefore, $\mathscr{C}$ is MHDR over $\mathcal{R}$ if and only if $\mathscr{C}_{\rho-1}$ is MHDR over $\mathbb{F}_q$.
\end{proof}
Examples illustrating the above results are given below.
\begin{example}
    Consider $\mathcal{R}=Z_{81}$ with the maximal ideal $\langle\gamma\rangle=\langle3\rangle$ and nilpotency index $\rho=4$. The corresponding residue field is $\mathbb{F}_3$. Let $\mathscr{C}$ be a two-dimensional $(80,28)$-constacyclic code of length $14$ over $Z_{81}$ i.e $\mathscr{C}$ is an ideal of $Z_{81}[x,y]/\langle x^7-80, y^2-28 \rangle$. The primitive idempotents of $Z_{81}[y]/\langle y^2-28 \rangle$ are $\psi_0(y)=41+67y,~ \psi_1(y)=41+14y$. Consider $80$-constacyclic codes of length $7$ over $Z_{81}$ as $J_0=\langle3\rangle, ~ J_1=\langle1+x\rangle$. By theorem~\ref{thm:structure}, the set $\{3(41+67y), (1+x)(41+14y) \}$ generates a two-dimensional $(80,28)$-constacyclic code of length $14$ over $Z_{81}$. $\mathcal{C}_{\rho-1}=\langle2+y,(x+1)(2+2y)\rangle$ is a two-dimensional $(2,1)$-constacyclic code of length $14$ over $\mathbb{F}_3$.
    The minimum hamming distance of this code is $2$. By theorem~\ref{thm:mhdr}, $\mathscr{C}$ is MHDR code over $Z_{81}$.
\end{example}
    
\begin{example}
    Consider $\mathcal{R}=Z_{361}$ with the maximal ideal $\langle\gamma\rangle=\langle19\rangle$ and nilpotency index $\rho=2$. The corresponding residue field is $\mathbb{F}_{19}$. Let $\mathscr{C}$ be a two-dimensional $(7,69)$-constacyclic code of length $18$ over $Z_{361}$ i.e., $\mathscr{C}$ is an ideal of $Z_{361}[x,y]/\langle x^6-7, y^3-69 \rangle$. The primitive idempotents of $Z_{361}[y]/\langle y^3-69 \rangle$ are $\psi_0(y)=241+343y+250y^2,~ \psi_1(y)=241+159y+33y^2,~ \psi_2(y)=241+220y+78y^2$. Consider $7$-constacyclic codes of length $6$ over $Z_{361}$ as $J_0=\langle x-2\rangle, ~ J_1=\langle19,(x-5)\rangle, ~J_2=\langle x^2-5x+6\rangle$. By theorem~\ref{thm:structure}, the set $\{(x-2)(241+343y+250y^2),19(241+159y+33y^2),)(x-5)(241+159y+33y^2),(x^2+5x+6)(241+220y+78y^2) \}$ generates a two-dimensional $(7,69)$-constacyclic code of length $18$ over $Z_{361}$ with $Rank(\mathcal{C})=15$. $\mathcal{C}_{\rho-1}=\langle(x-2)(13+y+3y^2),(13+7y+14y^2),(x^2-5x+6)(13+11y+2y^2)\rangle$ is a two-dimensional $(7,12)$-constacyclic code of length $18$ over $\mathbb{F}_{19}$. The minimum hamming distance for this code is $3$. By theorem~\ref{thm:mhdr}, $\mathscr{C}$ is not MHDR code over $Z_{361}$.
\end{example}

\section{Conclusion}
Firstly, we determined the primitive idempotents of the ring $\mathcal{R}[y]/ \langle y^{\mathrm{m}}-\mu \rangle$. Subsequently, using the structure of constacyclic codes of length $\ell$ over $\mathcal{R}$ and the primitive idempotents obtained, we established the generators of two-dimensional $(\lambda, \mu)$-constacyclic codes of length $\ell\mathrm{m}$ over $\mathcal{R}$ whose residue field $\mathbb{F}_q$ satisfies $q \equiv 1 \pmod{r\mathrm{m}}$, where $r$ denotes the multiplicative order of $\bar{\mu}$. Finally, we found a condition under which such codes are MHDR over $\mathcal{R}$. Some examples are given throughout the paper to illustrate the obtained results.

\printbibliography

\end{document}